\documentclass[a4paper,UKenglish,cleveref, autoref,thm-restate]{lipics-v2021}

\pdfoutput=1
\hideLIPIcs
\nolinenumbers

\title{The Rezk Completion for Elementary Topoi}

\author{Kobe Wullaert}{Delft University of Technology, Netherlands \and \url{https://kobewullaert.github.io}}{K.F.Wullaert@tudelft.nl}{https://orcid.org/0000-0003-4281-2739}{}
\author{Niels {van der Weide}}{Radboud University Nijmegen, Netherlands \and \url{https://nmvdw.github.io}}{nweide@cs.ru.nl}{https://orcid.org/0000-0003-1146-4161}{This research was supported by the NWO project ``The Power of Equality'' OCENW.M20.380, which is financed by the Dutch Research Council (NWO).}

\authorrunning{K. Wullaert and N. van der Weide} 
\Copyright{Jane Open Access and Joan R. Public}



\category{}

\relatedversion{}

\acknowledgements{We thank Benedikt Ahrens for valuable feedback on the versions of this paper and for the many helpful discussions.
Furthermore, we thank Arnoud van der Leer for reviewing the formalization and for refactoring the material on cartesian closed categories.
We gratefully acknowledge the work by the Rocq development team in providing the Rocq proof assistant and surrounding infrastructure, as well as their support in keeping UniMath compatible with Rocq.}%

\usepackage{tikz-cd}

\newcommand{\NNO}{\ensuremath{\mathbb{N}}}
\newcommand{\id}{\ensuremath{\mathsf{id}}}
\newcommand{\ob}[1]{\ensuremath{\mathsf{ob}(#1)}}
\newcommand{\op}[1]{\ensuremath{{#1}^\mathsf{op}}}

\newcommand{\cat}[1]{\ensuremath{\mathcal{#1}}}
\renewcommand{\AA}{\cat{A}}
\newcommand{\BB}{\cat{B}}
\newcommand{\CC}{\cat{C}}
\newcommand{\DD}{\cat{D}}
\newcommand{\EE}{\cat{E}}

\newcommand{\ccat}[1]{\mathsf{#1}}
\newcommand{\CAT}{\ccat{Cat}}
\newcommand{\SET}{\ccat{Set}}
\newcommand{\MON}{\ccat{Mon}}

\newcommand{\LCAT}{\ccat{FinLim}}
\newcommand{\CLCAT}{\ccat{FinColim}}
\newcommand{\OLCAT}{\LCAT_\Omega}

\newcommand{\NNOCAT}{\CAT_\NNO}
\newcommand{\CCC}{\ccat{CCC}}
\newcommand{\TOP}{\ccat{ElTop}}

\newcommand{\disp}[1]{\overline{#1}}
\newcommand{\XX}{x}
\newcommand{\YY}{y}

\newcommand{\xxx}{\disp{x}}
\newcommand{\yyy}{\disp{y}}
\newcommand{\zzz}{\disp{z}}
\newcommand{\ff}{f}
\newcommand{\fff}{\disp{f}}
\renewcommand{\gg}{g}
\renewcommand{\ggg}{\disp{g}}

\newcommand{\RC}{\mathsf{RC}}
\newcommand{\univ}[1]{{#1}_{\mathsf{univ}}}

\newcommand{\inclcat}{\iota}
\newcommand{\jota}{\hat{\inclcat}}
\newcommand{\iotaiota}{\bar{\inclcat}}

\newcommand{\ev}{\ensuremath{\mathsf{ev}}}
\newcommand{\expo}[2]{\ensuremath{{#2}^{#1}}}

\newcommand{\UU}{\ensuremath{\mathcal{U}}}
\newcommand{\AAA}{\ensuremath{A}}

\newcommand{\PER}{\ensuremath{\mathsf{PER}}}

\newcommand{\idtoiso}{\mathsf{idtoiso}}

\newcommand{\trunc}[1]{\ensuremath{\| {#1} \|}}

\newcommand{\dfn}[1]{\textbf{#1}}
\newcommand{\eg}{e.\,g.\,}

\newcommand{\UniMath}{\textsf{UniMath}}
\newcommand{\Coq}{\textsf{Rocq}}

\newcommand{\coqsymbol}{\begingroup\normalfont\includegraphics[height=\fontcharht\font`\{]{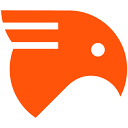}\endgroup}
\newcommand{\coqdocbaseurl}{https://anonymousdocs-cb.github.io/CSL2026_unimath_doc/HtmlDoc_94b49c3/UniMath.}
\newcommand{\coqident}[2]{\href{\coqdocbaseurl #1.html\##2}{\coqsymbol}}
\newcommand{\coqidenturl}[3]{\href{\coqdocbaseurl #1.html\##3}{\coqsymbol}}
\newcommand{\coqfile}[1]{\href{\coqdocbaseurl #1.html}{\nolinkurl{#1}}}

\newcommand{\shortcommitnumber}{94b49c3}

\EventEditors{John Q. Open and Joan R. Access}
\EventNoEds{2}
\EventLongTitle{42nd Conference on Very Important Topics (CVIT 2016)}
\EventShortTitle{CVIT 2016}
\EventAcronym{CVIT}
\EventYear{2016}
\EventDate{December 24--27, 2016}
\EventLocation{Little Whinging, United Kingdom}
\EventLogo{}
\SeriesVolume{42}
\ArticleNo{23}

\begin{document}

\maketitle

\begin{abstract}
The development of category theory in univalent foundations and the formalization thereof is an active field of research.
Categories in that setting are often assumed to be univalent which means that identities and isomorphisms of objects coincide.
One consequence hereof is that equivalences and identities coincide for univalent categories and that structure on univalent categories transfers along equivalences.
However, constructions such as the Kleisli category, the Karoubi envelope, and the tripos-to-topos construction, do not necessarily give univalent categories.
To deal with that problem, one uses the Rezk completion, which completes a category into a univalent one.
However, to use the Rezk completion when considering categories with structure, one also needs to show that the Rezk completion inherits the structure from the original category.

In this work, we present a modular framework for lifting the Rezk completion from categories to categories with structure.
We demonstrate the modularity of our framework by lifting the Rezk completion from categories to elementary topoi in manageable steps.
\end{abstract}

\section{Introduction}

Univalent foundations \cite{hottbook,UPbook} has proven itself to be an adequate foundation for the development of mathematics and category theory in particular.
In such a foundation, structures adhere to the principle that equality and equivalence of structures coincide, such that equivalent structures are indistinguishable.
This principle is known as the univalence principle \cite{UPbook}.

For set-level structures, this principle is a \emph{direct} consequence of the univalence axiom, which postulates that principle for types.
However, for higher-dimensional structures, such as categories, the story is more subtle \cite{AKS2015,hottbook}.
For categories to adhere to this principle, we need an equivalence between the identity type $(\AA = \CC)$ and the type $(\AA \simeq \CC)$ of \emph{equivalences}.
The univalence axiom implies that we have an equivalence between $(\AA = \CC)$ and the type of \emph{isomorphisms} $(\AA \cong \CC)$.
Hence, the univalence principle for categories can be restated as having an equivalence $(\AA \cong \CC) \simeq^{(1)} (\AA \simeq \CC)$.
However, $(1)$ does not hold in general, unless when we consider \emph{univalent categories}: categories for which identity and isomorphism for objects are equivalent.

The univalence condition for categories is met by many naturally occurring categories and univalent categories correspond to categories in Voevodsky's simplicial set model of homotopy type theory \cite{KL2021}.
Nevertheless, different constructions do not produce univalent categories; see \cref{sec:on-rezk-for-cats}.
This is where the \emph{Rezk completion} comes into play.

Abstractly, the Rezk completion refers to the univalent completion of a category \cite{AKS2015}.
That is, the Rezk completion of a category $\CC$ is a univalent category $\RC(\CC)$, together with a functor $\CC \xrightarrow{\eta_\CC} \RC(\CC)$, which satisfies the following universal property.
For every univalent category $\EE$, any functor of type $\CC \to \EE$ factors uniquely through $\eta_\CC$.
In \cite{AKS2015}, it was shown that this universal property is equivalent to the requirement that $\eta_\CC$ is a \emph{weak} equivalence.
However, the weakness of this equivalence means that $\eta_\CC$ need not have an inverse.
In particular, structure on $\CC$ is not guaranteed to transport across $\eta_\CC$ to a structure on $\RC(\CC)$.

In this work, we present a modular approach to lift the Rezk completion from categories to structured categories, such as elementary topoi.
In particular, we show that the Rezk completion of a category inherits the structure of a topos and that the universal property extends to the level of topoi.
Furthermore, this work is accompanied by a formalization in the \Coq -library \UniMath, and our approach is designed with this formalization in mind.
In particular, our approach is modular and is designed to be applicable to a lot of structures a category can have.
Furthermore, we expect our approach to be practical for the formalization hereof in any proof assistant based on homotopy type theory.

\subsection{Prior Work}

As for ($1$-)categories, higher-dimensional categories are supposed to satisfy an additional univalence condition such that they adhere to the univalence principle.
A general theory making this precise has been developed in \cite{UPbook} where a signature for finite-dimensional categorical structures has been developed and a notion of equivalence (indiscernibility) has been formulated for such structures, which for e.g. categories coincides with the notion of equivalence.

Different categorical theories have been developed in univalent foundations, including the theory of monoidal categories \cite{WMA2022}, enriched categories \cite{vdW24-enriched}, bicategories \cite{AFMVW2021}, and double (bi)categories \cite{WRAN2024,RWAN2025}.
In particular, those works formulated the notion of univalence for such categorical structures, which coincide with the one given by \cite{UPbook}.

The Rezk completion \cite{AKS2015} has been lifted to multiple categorical structures, including some of the mentioned above \cite{AFMVW2021,WMA2022,vdW24-enriched}.
Furthermore, the Rezk completion has been constructed for structured groupoids in \cite{VW2021}, where the structure is defined via higher inductive types.

Moreover, univalent category theory, and the Rezk completion in particular, has been used in the study of denotational semantics.
For instance, the Rezk completion is used to transform a category with representable maps of presheaves into a univalent category with families \cite{ALV2017}.
Furthermore, the Rezk completion is used to relate two different constructions of Scott's representation theorem, which model the untyped $\lambda$-calculus \cite{LWA2025}.
Lastly, various classes of topoi and their internal logic, based on comprehension categories, are constructed in the univalent setting in \cite{vdW2024-internallang-univcat}.

\subsection{Overview and Contributions}

This paper is structured as follows:

\begin{enumerate}
\item 
	\cref{sec:on-rezk-for-cats} contains background material on univalent categories and the Rezk completion.
\item 
	In \cref{sec:on-lifting-adj}, we discuss the framework and strategy for lifting the Rezk completion from categories to structured categories.
	We start by briefly recalling some of the basics of displayed bicategories.
	Then we define the notion of displayed universal arrows as an intermediate step for the modular construction of biadjoints.
	Afterwards, we discuss how such arrows specialize in the case of lifting the Rezk completion.
\item 
	In \cref{sec:on-rezk-for-top}, we apply the framework presented in \cref{sec:on-lifting-adj} to lift the Rezk completion from categories to various classes of structured categories where we work our way up to W-topoi.
	More precisely, we consider finitely co/complete categories (\cref{sec:RC-completeness}), cartesian closed closed categories (\cref{sec:RC-CCC}), elementary topoi (\cref{sec:RC-TOP}), and W-topoi (\cref{sec:RC:TOPN}).
\item 
	In \cref{sec:related,sec:conclusion}, we elaborate on the different methods to extend the Rezk completion to categories with structure found in the literature and possible ways in which this work can be extended.
\end{enumerate}

\subsection{Formalization and Foundations}

This work is accompanied by formalization in the \Coq -library \UniMath on univalent mathematics \cite{Coq-refman,UniMath}.
The logic underlying the \UniMath library is an intensional dependent type theory whose universes satisfy the univalence axiom.
We do not rely on the axiom of choice, except in \cref{lemma:AC-RC-inf}, nor the law of excluded middle.
We rely on propositional truncation.

The formalization, and this work, makes heavily use of the already existing material on (bi)category theory in the \UniMath -library and the accompanied literature \cite{AKS2015,AFMVW2021}.
The environments annotated with a \coqident{}{} denote that the result is formalized.
Upon clicking \coqident{}{}, the reader is directed towards the HTML documentation.
The commit number from which the documentation is generated is \shortcommitnumber.

\section{On the Rezk Completion for Categories}
\label{sec:on-rezk-for-cats}

In this section, we recall and elaborate on the Rezk completion for categories.
Most of the material in this section is already found in \cite{AKS2015}.
First, we recall the definition of univalent categories in \cref{dfn:univ-cat}.
Then we recall weak equivalences (\cref{dfn:weak-equiv}) and we recall the Rezk completion in terms of weak equivalences (\cref{dfn:Rezk-Completion}).
Afterwards, we zoom in on some examples of categorical constructions where the Rezk completion plays a role.
Lastly, we discuss a bicategorical formulation of the Rezk completion, which we use in \cref{sec:on-rezk-for-top,sec:on-lifting-adj} to lift the Rezk completion from categories to structured categories.

\begin{subparagraph}{Univalent Categories}
In univalent foundations, \textbf{categories} $\CC$ consist in particular of a type of objects $\ob{\CC}$ and for every pair of objects $x,y : \ob{\CC}$ a type of morphisms $\CC(x,y)$.
Furthermore, a category $\CC$ consists of a composition operation on the morphisms $(- \cdot -) : \CC(x,y) \to \CC(y,z) \to \CC(x,z)$, identity morphisms $\id_x : \CC(x,x)$,
and the composition is associative and unital.
Lastly, each $\CC(x,y)$ is a set, which means that equalities between morphisms form a proposition; recall that a type $A$ is a proposition if each identity type $x = y$ is equivalent to the unit type, for every $x,y : A$.

However, as mentioned in the introduction, such categories do not adhere to the principle that equality and equivalence coincide.
That is, the univalence axiom implies that equality and isomorphism for categories coincide \cite[Lemma 6.16]{AKS2015}, but equality and isomorphism do not coincide with equivalence for categories.
Those categories for which equivalences do coincide with the aforementioned are precisely the \textbf{univalent categories}.
These are categories for which identities and isomorphisms between two objects coincide.
More precisely, let $\CC$ be a category and denote by $(x \cong y)$ the type of isomorphisms from $x$ to $y$, for $x,y : \ob{\CC}$.
By path induction, we have a function $\idtoiso_{x,y} : (x = y) \to (x \cong y)$.
\begin{definition}
[\coqident{CategoryTheory.Core.Univalence}{is_univalent} Dfn 3.6. \cite{AKS2015}]
\label{dfn:univ-cat}
A category $\CC$ is \textbf{univalent} if $\idtoiso_{x,y}$ is an equivalence of types, for every $x, y : \ob{\CC}$.
\end{definition}
By the univalence axiom, many categories, such as those of sets, groups, etc., are univalent.

Observe that univalence implies that $\ob{\CC}$ is a groupoid, i.e., the types $x=y$ are sets.
\end{subparagraph}

\begin{subparagraph}{The Rezk completion}
\label{para:Rezk-completion}
The Rezk completion refers to the univalent completion of a category \cite{AKS2015,hottbook}.
That is, if $\CC$ is a category, then its Rezk completion is a univalent category $\hat{\CC}$ together with a functor $\eta : \CC \to \hat{\CC}$ which is initial among the functors from $\CC$ to a univalent category, i.e., for every $\EE$ univalent, the precomposition functor $(\eta \cdot -) : [\hat{\CC}, \EE] \to [\CC, \EE]$ is an isomorphism of categories, where $[-,-]$ denotes the functor category.
In particular, this means that every functor into a univalent category factors uniquely along through $\eta$:
\[
\begin{tikzcd}
\CC \arrow[rr, "\eta"] \arrow[rd, swap, "\forall F"] & & \hat{\CC} \arrow[dl, dashed, "\exists! H"] \\
& \EE &
\end{tikzcd}
\]

We start by recalling weak equivalences, which allow us to describe the universal property of the Rezk completion.

\begin{definition}
[\coqident{CategoryTheory.WeakEquivalences.Core}{is_weak_equiv} Dfn 6.7. \cite{AKS2015}]
\label{dfn:weak-equiv}
Let $\CC, \EE$ be categories.
A \textbf{weak equivalence} is a functor $F : \CC \to \EE$ which is fully faithful and is \textbf{essentially surjective} (on objects); 
where essential surjectivity (\coqident{CategoryTheory.Core.Functors}{essentially_surjective}) means that for every object $y : \EE$ in the codomain, there merely exists an object $x : \CC$ and an isomorphism of type $F\,x \cong y$.
\end{definition}

Observe that essential surjectivity does not imply that for every $y : \EE$ one can choose an object $x : \CC$ in the preimage, due to the truncation in the mere existence.
In particular, there need not to be a function $F^{-1} : \ob{\EE} \to \ob{\CC}$ such that $F(F^{-1}(y)) \cong y$.
Those functors $F : \CC \to \EE$ for which for every $y : \ob{\EE}$ we can choose an object $x : \ob{\CC}$ such that $F\,x \cong y$ are said to be \dfn{split surjective} (\coqident{CategoryTheory.Core.Functors}{split_essentially_surjective}).
Furthermore, being split surjective and fully faithful is equivalent to $F$ having a right adjoint whose unit and counit are isomorphisms.
Hence, such $F$'s are called \dfn{(adjoint) equivalences}.
However, as opposed to weak equivalences, being an equivalence is generally not a proposition.
In particular, a weak equivalence is not an equivalence.
In fact, between categories whose type of objects are sets (setcategories), the statement that weak equivalences are split surjective is equivalent to the axiom of choice (\cite[Page 309]{hottbook}).
However, between univalent categories, weak equivalences, equivalences, and even isomorphisms (\cite[Dfn. 6.9, Lemma 6.10]{AKS2015}) do coincide.
See \cite{AKS2015} for more details.

\begin{proposition}[\coqident{CategoryTheory.WeakEquivalences.Core}{precomp_is_iso} Thm 8.4. \cite{AKS2015}]
\label{prop:precompiso}
Let $G : \CC_0 \to \CC_1$ be a weak equivalence between not-necessarily univalent categories.
Then for every univalent category $\CC_2$, the precomposition functor $(G \cdot -) : [\CC_1, \CC_2] \to [\CC_0, \CC_2]$
is an equivalence of univalent categories.
\end{proposition}
For a characterization of univalence in terms of weak equivalences, see \cite[Thm. 8.9]{AKS2015}.

The Rezk completion \cite{AKS2015} refers to a construction of the univalent completion (see \cref{rem:existence-RC}).
Nevertheless, we only rely on the description of the Rezk completion in terms of weak equivalences.
Hence, we will refer to \emph{the Rezk completion} as in the following definition:
\begin{definition}
\label{dfn:Rezk-Completion}
Let $\CC$ be a category.
A \textbf{Rezk completion} for $\CC$ consists of a univalent category $\hat{\CC}$ and a weak equivalence $\eta : \CC \to \hat{\CC}$.
\end{definition}

That a Rezk completion indeed satisfies the universal property of the univalent completion is precisely \cref{prop:precompiso}.
Furthermore, for every category, its type of Rezk completions is a proposition.
In particular, Rezk completions are unique up to identity, and hence, we can say \textit{the} Rezk completion.
Given a category $\CC$, we denote by the $\RC(\CC) := \hat{\CC}$ the univalent category and $\eta_\CC : \CC \to \RC(\CC)$ the weak equivalence.
\end{subparagraph}

\begin{remark}
\label{rem:existence-RC}
The Rezk completion has two constructions \cite{AKS2015,hottbook}.
In \cite{AKS2015}, the Rezk completion of a category $\CC$ is constructed as a full subcategory of the category of presheaves $[\op{\CC}, \SET]$ given by those functors which are merely isomorphic to a representable functor.
However, this construction raises the universe level since $[\op{\CC}, \SET]$ has a higher universe level than $\CC$.
More precisely, recall that there is a hierarchy of universes, let $\mathcal{U}_k$ be the universe at level $k$, and denote by $\CAT^{(i,j)}$ the type of categories whose type of objects is in $\mathcal{U}^{(i)}$ and where every $\CC(x,y)$ is in $\mathcal{U}_j$.
Whenever we have $\CC : \CAT^{(i,j)}$, we have $\RC(\CC) : \CAT^{(i \vee (j+1), i \vee j)}$, where $\vee$ is the least upperbound of universe levels.

We can also consider a second construction \cite{hottbook}, assuming the existence of higher inductive types in the same universe.
Here, we can define the type of objects of $\RC(\CC)$ as a higher inductive type which contains, in particular, a term constructor $i : \ob{\CC} \to \ob{\RC(\CC)}$ and a path constructor $\prod_{x y : \CC} (x \cong y) \to (i\,x = i\,y)$.
This construction has the advantage that the Rezk completion remains in the same universe.
\end{remark}

\begin{subparagraph}{Examples}
\label{para:examples-UCRC}
We now discuss some examples of categories and constructions thereof to illustrate the following aspects concerning the Rezk completion.
\cref{exa:set-cats} illustrates that the Rezk completion generalizes certain constructions on set-theoretic structures such as groups and posets.
\cref{exa:kleisli,exa:idemp} illustrate that working in univalent foundations distinguishes constructions which are classically equivalent.
\cref{exa:trip-to-top} is the motivating example of this work, where we have a construction of topoi which usually does not produce a univalent topos.
Hence, to construct an alternative univalent construction thereof, the Rezk completion is essential.
In particular, \cref{exa:trip-to-top} illustrates that we need sufficient results about the transfer of structure from a category to its Rezk completion.
Of course, this observation also plays a role in the other examples.

\begin{example}[Setcategories]
\label{exa:set-cats}
Setcategories are categories whose type of objects is a set.
Often, such categories are not univalent. 

\begin{enumerate}
\item The walking iso category, i.e., the category with two distinct but isomorphic objects, is not univalent.
The Rezk completion hereof is the terminal category.
More generally, we can consider sets with a preorder, which are equivalently setcategories whose hom-sets are propositions.
Such a category is univalent if and only if the preorder is anti-symmetric \cite[Exa. 3.14]{AKS2015}
Hence, the Rezk completion of a preorder is its posetal completion.
\item A monoid $M$ is equivalently a one-object setcategory $BM$.
If $M$ contains more than one invertible element, this category is not univalent.
If $M$ is an abelian group, the Rezk completion of $BM$ is the Eilenberg-MacLane space of $M$ \cite{LF_EMspaces_in_HoTT,BDR_highergroupsinHoTT}, which in univalent foundations is constructed via higher-inductive types.
In particular, whereas the Rezk completion for preordered sets can be understood as picking out a representative of every equivalence class, this example illustrates that non-trivial identities are added.
Furthermore, $\eta : BM \to \RC(BM)$ is an example of a weak equivalence that does not have an inverse.
\item The Rezk completion allows for the construction of quotients of sets.
Indeed, let $\AAA$ be a set equipped with an equivalence relation $\simeq$ (i.e., a setoid).
Then $\AAA$ can be made into a category whose objects are the elements in $\AAA$ and a morphism of type $a \to b$ witnesses $a \sim b$.
This makes $\AAA$ into a groupoid.
The Rezk completion for $\AAA$ is then a groupoid whose objects are equivalence classes of $\sim$.
More generally, the Rezk completion can be seen as a generalization of the groupoid quotient \cite{VW2021}.
\end{enumerate}
Moreover, many free categories, as generated by a syntax or a graph, are setcategories.
\end{example}

\begin{example}[Kleisli category]
\label{exa:kleisli}
Monads are an important tool in functional programming as they provide the right abstraction to model programs with effects.
Intuitively, given a category $\CC$ of datatypes and programs, and a monad $T$ which specifies an effect, its Kleisli category contains the same objects (datatypes), but the morphisms are programs where the effect is integrated.

The Kleisli category of $T$ can be constructed via Kleisli morphisms or via free algebras.
While the latter construction produces a univalent category, if $\CC$ is univalent, the former does not \cite{W2025,UPbook}.
However, the Rezk completion of the Kleisli morphism construction is precisely the free algebra construction.
\end{example}

\begin{example}[idempotent completion]
\label{exa:idemp}
The (split-)idempotent completion is a fundamental construction in many areas.
For instance, in category theory, this completion can be used to study equivalences between presheaf categories \cite[Thm 6.5.11]{Borceux1994handbook1}.
Furthermore, this construction generalizes the Cauchy completion of metric spaces \cite{lawvere_metric_1973}.
Moreover, the idempotent completion can be used to construct models of the untyped lambda-calculus \cite{Hyland2017}.
We discuss two constructions of the idempotent completion: the \emph{Cauchy completion} and the \emph{Karoubi envelope}.

To construct the Karoubi envelope of a category $\CC$, we use the notion of split idempotents.
A morphisms $f : \CC(x,x)$ is \emph{idempotent} if $f \cdot f = f$, and an idempotent morphism $f$ is \emph{split} if $f$ can be factorized as a composite $x \xrightarrow{\pi} y \xrightarrow{\iota} x$, for some $y : \CC$, and such that $\id_y = \iota \cdot \pi$.
The objects of the Karoubi envelope are the idempotent morphisms in $\CC$.

In general, the \emph{Karoubi envelope} is not univalent.

The Cauchy completion is constructed as a full subcategory of the category of presheaves.
The Cauchy completion is always univalent, and is precisely the Rezk completion of the Karoubi envelope.
See \cite{LWA2025} for more details.
\end{example}

\begin{example}[tripos-to-topos]
\label{exa:trip-to-top}
The tripos-to-topos construction is fundamental in topos theory \cite{Pitts2002,HJP1980}. One can use it to construct realizability topoi \cite{OostenRealizability} and to give an alternative description of categories of sheaves on a complete Heyting algebra [4]. Intuitively, both triposes and topoi represent models of higher-order intuitionistic logic, but there is a key difference. While one represents formulas as morphisms into the subobject classifier in a topos, triposes offer more flexibility. More specifically, in a tripos we have a functor $\mathbb{P} : \op{\SET} \to \mathsf{HA}$ that assigns to each set a Heyting algebra of formulas. In essence, the notion of formula in a topos $\mathcal{E}$ is internal to $\mathcal{E}$, but this is not so in a tripos. The tripos-to-topos construction constructs from every tripos $\mathbb{P} : \op{\SET} \to \mathsf{HA}$ a corresponding topos $\PER(\mathbb{P})$.

To understand the tripos-to-topos construction better, we consider an example. Every complete Heyting algebra $\mathcal{H}$ gives rise to a tripos such that the formulas on a set $X$ are the same as functions $X \to \mathcal{H}$. If we apply the tripos-to-topos construction, then we get the category of H-valued sets \cite{Higgs1984}, which is not necessarily univalent.

In general, the tripos-to-topos construction does not give a univalent category.
\end{example}

\begin{example}[functor categories]
\label{exa:functor-cats}
Let $\AA$ and $\CC$ be categories.
Univalence is preserved under taking functor categories.
That is, if $\CC$ is univalent, then the category $[\AA,\CC]$ of functors is also univalent \cite[Thm. 4.15]{AKS2015}.
However, even though the functor categories preserve univalence, they do not preserve the Rezk completion.
That is, the Rezk completion of a functor category $[\AA,\CC]$ is not the functor category $[\AA, \RC(\CC)]$.
Indeed, the post-composition functor $(- \cdot \eta_\CC) : [\AA, \CC] \to [\AA, \RC(\CC)]$ is fully faithful but not essentially surjective.
However, the Rezk completion of $[\AA, \CC]$ can be constructed as the replete image of $(- \cdot \eta_\CC)$.
\end{example}

\end{subparagraph}

\begin{subparagraph}{A bicategorical perspective}
To lift the Rezk completion from categories to structured categories, we formulate the Rezk completion bicategorically.

Analogous to the way that structured sets, e.g., monoids, form the objects in their respective categories, categories with additional structure form the objects in a bicategory \cite{AFMVW2021}.
One advantage of working bicategorically is that we can treat structure uniformly.
Furthermore, the universal property of the Rezk completion can be internalized in the bicategory $\CAT$ of categories, functors, and natural transformations, since the Rezk completion makes the bicategory $\univ{\CAT}$ of univalent categories into a reflective subbicategory of $\CAT$.
That is, the inclusion of $\univ{\CAT}$ into $\CAT$ has a left biadjoint (\cref{cor:rezk-as-leftadj}).

\begin{definition}
[\coqident{Bicategories.PseudoFunctors.UniversalArrow}{left_universal_arrow}]
\label{dfn:pseudofunctor}
A pseudofunctor $R : \BB_1 \to \BB_2$ has a \textbf{left biadjoint} $L$ if
we have a function $L : \BB_2 \to \BB_1$, a family of morphisms $\eta : \prod_{X : \BB_2} X \to R(L(X))$ called \textbf{the unit}, and for every $X : \BB_2$ and $Y : \BB_1$, the functor $\eta_X \cdot R(-) : \BB_1(L\,X, Y) \to \BB_2(X, R\,Y)$ is an adjoint equivalence of categories.
We write $L \dashv R$ to denote that $L$ is a left biadjoint of $R$.
\end{definition}

\begin{proposition}
\label{cor:rezk-as-leftadj}
The inclusion of $\univ{\CAT}$ into $\CAT$ has a left biadjoint $\RC : \CAT \to \univ{\CAT}$.
In particular, the action on objects is given by $\CC \mapsto \RC(\CC)$, and the unit is pointwise given by $\eta_\CC$.
\end{proposition}

In \cref{sec:disp-biadj,sec:on-rezk-for-top}, we elaborate on how the theory of (displayed) bicategories, in combination with \cref{cor:rezk-as-leftadj}, can be used for modularly constructing the Rezk completion for categories with structure.
\end{subparagraph}

\section{On the Lifting of Biadjoints}
\label{sec:on-lifting-adj}

The examples presented in \cref{sec:on-rezk-for-cats} illustrate that various constructions for categories often found produce a not-necessarily univalent category.
The Rezk completion provides a solution to this problem. 
However, while structure on a category can be transported along isomorphisms, we need to verify such transportations for weak equivalences by hand.

Now we can describe the Rezk completion and its universal property a left biadjoint.
That is, the Rezk completion forms a pseudofunctor $\RC : \CAT \to \univ{\CAT}$, which is a left biadjoint to the forgetful pseudofunctor $\iota : \univ{\CAT} \to \CAT$.
The transportation of structure from categories to their Rezk completions can then be described as follows.
Consider a structure on categories and let $\BB$ be the bicategory of categories equipped with that structure and whose morphisms preserve the structure.
For example, consider categories with finite products.
Then $\BB$ is the bicategory of finite products and functors that preserve those products.
We have an inclusion $\bar{\iota} : \univ{\BB} \hookrightarrow \BB$, where $\BB$ is the bicategory of univalent categories equipped finite products.
By forgetting the products, we have forgetful pseudofunctors $\mathsf{U} : \BB \to \CAT$ and $\mathsf{U}\vert_{\mathsf{univ}} : \univ{\BB} \to \univ{\CAT}$.
Furthermore, we have that the following diagram commutes:
\[
\begin{tikzcd}
\univ{\BB} \arrow[d, "\mathsf{U}\vert_{\mathsf{univ}}", swap] \arrow[r, "\bar{\iota}"] & \BB \arrow[d, "\mathsf{U}"] \\
\univ{\CAT} \arrow[r, "\iota"] & \CAT
\end{tikzcd}
\]
The transportation and preservation of finite products by the Rezk completion then corresponds to having a left biadjoint $\bar{\RC}$ of $\bar{\iota} : \univ{\BB} \hookrightarrow \BB$.
We say that the Rezk completion \dfn{lifts} from categories to categories equipped with that structure if a left biadjoint $\bar{\RC}$ exists and if the following diagram commutes:
\[
\begin{tikzcd}
\univ{\BB} \arrow[d, "\mathsf{U}\vert_{\mathsf{univ}}", swap] & \BB \arrow[d, "\mathsf{U}"] \arrow[l, "\bar{\RC}", swap] \\
\univ{\CAT} & \CAT \arrow[l, "\RC"].
\end{tikzcd}
\]

In this section, we present a framework and strategy for constructing such lifts, which we use in \cref{sec:on-rezk-for-top} to lift the Rezk completion from categories to elementary topoi.
The methodology builds on the theory of displayed bicategories \cite{AFMVW2021}.
In \cref{sec:disp-bicat-structured-cats}, we discuss how we represent structure on categories via displayed bicategories.
Furthermore, we introduce displayed universal arrows in \cref{sec:disp-biadj}, which allows us to lift the Rezk completion from categories to structured categories.
In \cref{sec:disp-RC}, we observe that the lifting to e.g., elementary topoi often simplifies in at least two ways, which leads to \cref{prop:disp-univ-arrow-over-cat}.

\subsection{Displayed Bicategories and Structured Categories}
\label{sec:disp-bicat-structured-cats}

In the remainder of this paper, we make heavy use of displayed bicategories,
and we start by recalling this notion.
Displayed bicategories serve various purposes, and one of these is that we can use them to modularly construct bicategories and to modularly prove their univalence \cite{AFMVW2021,AL2019}.
To understand what displayed bicategories are, let us first recall an example of a displayed category.
We can construct the category $\MON$ of monoids and monoid homomorphisms by endowing the objects and morphisms of $\SET$ with extra structure and properties.
Specifically, for every set there is a type of monoid structures on it, and for every function between sets with monoid structures on them we have a type expressing that this function is a homomorphism.
One must also prove that the identity function is a homomorphism and that homomorphisms are closed under composition.
Displayed categories generalize such descriptions of categories.
Specifically, in a displayed category $\DD$ over a category $\CC$,
we have a type $\DD(x)$ of displayed objects over every object $x : \CC$ and a set $\bar{x} \to_{f} \bar{y}$ of displayed morphisms for every morphism $f : x \rightarrow y$ and displayed objects $\bar{x} : \DD(x)$ and $\bar{y} : \DD(y)$.
Displayed bicategories generalize displayed categories to the bicategorical setting (also see \cref{rem:dfn-bicat}).

\begin{definition}%
[\coqident{Bicategories.DisplayedBicats.DispBicat}{disp_bicat}]
\label{dfn:disp-bicat}
Let $\BB$ be a bicategory.
A \dfn{displayed bicategory} (with contractible $2$-cells) $\DD$ over $\BB$ consists of
\begin{enumerate}
  \item for every $\XX : \BB$, a type $\DD(\XX)$;
  whose terms are called \textit{displayed objects};
  \item for every $\ff : \BB(\XX, \YY)$, $\xxx : \DD(\XX)$, and $\yyy : \DD(\YY)$,
    a type $\xxx \rightarrow_\ff \yyy$;
  whose terms are called \textit{displayed morphisms};
  \item for every $\XX : \BB$ and $\xxx : \DD(\XX)$, a displayed morphism of type $\xxx \rightarrow_{\id_\XX} \xxx$;
  \item for every $\fff : \xxx \rightarrow_\ff \yyy$ and $\ggg : \yyy \rightarrow_\gg \zzz$, a displayed morphism of type $\xxx \rightarrow_{\ff \cdot \gg} \zzz$.
\end{enumerate}
\end{definition}

\begin{remark}
\label{rem:dfn-bicat}
The definition of a displayed bicategory as presented in \cref{dfn:disp-bicat} is a simplified version as compared to \cite{AFMVW2021}.
The reason for this simplification is that for the displayed bicategories we consider, e.g., elementary topoi, no compatibility conditions on the natural transformations need to be required.
That is, the $2$-cells in the bicategory of elementary topoi are natural transformations.

From now on, we leave implicit that the displayed bicategories have contractible $2$-cells.
In particular, the bicategorical laws need not be verified since these are phrased in terms of equality of $2$-cells.
For instance, the associativity of composition, up to an invertible $2$-cells, follows directly from the statement for categories.
\end{remark}

\begin{definition}
[\coqident{Bicategories.DisplayedBicats.DispBicat}{total_bicat}]
\label{dfn:total-bicat}
The \dfn{total bicategory} of a displayed bicategory $\DD$, denoted $\int \DD$, is the bicategory whose objects are (dependent) pairs $(\XX, \xxx)$, with $\XX : \BB$ and $\xxx : \DD(\XX)$.
The morphisms are pairs $(\ff, \fff)$ with $\ff : \BB(\xxx,\yyy)$ and $\fff : \xxx \rightarrow_\ff \yyy$.
The $2$-cells are the $2$-cells in $\BB$.
\end{definition}
\begin{definition}
[\coqident{Bicategories.PseudoFunctors.Examples.Projection}{pr1_psfunctor}]
Let $\DD$ be a displayed bicategory over $\BB$.
We have a (strict) pseudofunctor $U : \int \DD \to \BB$, called the \dfn{forgetful pseudofunctor}.
\end{definition}

\begin{example}
[\coqident{Bicategories.DisplayedBicats.Examples.CategoriesWithStructure.Topoi}{disp_bicat_elementarytopoi'} The displayed bicategory of elementary topoi]
\label{exa:disp-bicat-of-top}
Elementary topoi are categories with a lot of structure, such as finite limits, exponentials, and a subobject classifier (for a definition of those structures see \cref{sec:on-rezk-for-top}).
We construct the bicategory of elementary topoi and logical functors as a displayed bicategory over $\CAT$ in multiple steps.
First, we define the displayed bicategory of finite limits, whose displayed objects equip a category with chosen finite limits, and whose displayed morphisms say that a functor preserves those limits.
Then the bicategory $\LCAT$ of finitely complete categories is the total bicategory of that displayed bicategory.
Second, we define analogously the displayed bicategory $\OLCAT$, resp.\, $\CCC$, of subobject classifiers and exponentials respectively, whose base bicategory is $\LCAT$.
The bicategory $\TOP$ is then constructed as the total bicategory of the product (\cite[Definition 6.6]{AFMVW2021}) of $\OLCAT$ and $\CCC$ over $\LCAT$.
\end{example}

\begin{remark}[the $\sum$ construction]
The framework presented below is developed to work for displayed bicategories over $\CAT$.
However, the displayed bicategory of elementary topoi (\cref{exa:disp-bicat-of-top}) as constructed there is a displayed bicategory over $\int_\CAT \DD$.
Nonetheless, if $\EE$ is a displayed bicategory over $\int_{\CAT} \DD$, the total bicategory $\int_{\int_\CAT \DD} \EE$ is equivalently a bicategory $\int_\CAT \sum_\DD \EE$ displayed over $\CAT$ by applying the $\sum$ construction for displayed bicategories (\cite[Definition 6.6]{AFMVW2021}).
Hence, by the displayed bicategory of elementary topoi, we refer to the one over $\CAT$.
\end{remark}

\begin{example}[The displayed bicategory of univalent elementary topoi and the restriction of displayed bicategories]
Univalent topoi are those topoi whose underlying category is univalent.
Hence, the bicategory $\univ{\TOP}$ of univalent topoi can be constructed as a full subbicategory of $\TOP$ (\cref{exa:disp-bicat-of-top}).
One way to construct $\univ{\TOP}$ is thus by restricting the displayed bicategory defining $\TOP$ to $\univ{\CAT}$.
More generally, if $\DD$ is a displayed bicategory over $\CAT$, we can restrict $\DD$ to a displayed bicategory $\univ{\DD}$ over $\univ{\CAT}$ where the displayed objects and morphisms are those in $\DD$; see  \coqident{Bicategories.DisplayedBicats.Examples.DispBicatOnCatToUniv}{disp_bicat_on_cat_to_univ_cat}.
\end{example}

Below, we sometimes use the following terminology.
Let $\DD$ be a displayed bicategory.
Then we say a $\DD$-structure on a category $\CC$ is a $\DD$-displayed object over $\CC$.

\subsection{Displayed Biadjoints}
\label{sec:disp-biadj}

Previously, we discussed that structured categories are the objects in a displayed bicategory $\DD$ over $\CAT$.
Furthermore, $\DD$ can be restricted to a displayed bicategory $\univ{\DD}$ over $\univ{\CAT}$.
That is, $\int \univ{\DD}$ is the full subbicategory of $\int \DD$ whose objects are univalent categories with a $\DD$-structure.
Hence, we have the following diagram of bicategories.
\begin{center}
\label{eqn:diagram:lifted-biadj}
\begin{tikzcd}
\int {\univ{\DD}}
	\arrow[d, swap, "U"]
& {\int \DD} \arrow[d, "U"] \\
\univ{\CAT}
\arrow[r, bend right, "\iota"] 
& \CAT \arrow[l, bend right, "\RC"]
\end{tikzcd}
\end{center}

We now reduce the problem of lifting the biadjunction $\RC \dashv \iota$ to the level of displayed bicategories.
First, we recall the notion of displayed pseudofunctor \cite{AFMVW2021} in \cref{dfn:disp-psfunctor}.
Second, we define (left) displayed universal arrows \cref{dfn:disp_left_universal_arrow}, which is specialized to $\RC \dashv \iota$ in \cref{sec:disp-RC}.

Observe that we have a pseudofunctor $\iotaiota$ from $\int \univ{\DD}$ to $\int \DD$ given by the identity.
Indeed, $\univ{\DD}$ is $\DD$ restricted to a full subbicategory and $\iota$ is a pseudofunctor.
More generally, we have that pseudofunctors $\mathcal{F}$ between total bicategories can often be constructed by a pseudofunctor $\mathcal{G}$ between the base bicategories and a \textit{displayed pseudofunctor} over $\mathcal{G}$:
\begin{definition}
[\coqident{Bicategories.DisplayedBicats.DispPseudofunctor}{disp_psfunctor}
\coqident{Bicategories.DisplayedBicats.DispPseudofunctor}{total_psfunctor}]
\label{dfn:disp-psfunctor}
Let $F : \BB_1 \to \BB_2$ be a pseudofunctor and $\DD_i$ a displayed bicategory over $\BB_i$, for $i=1,2$.
A \dfn{displayed pseudofunctor} $\hat{F}$ over $F$ consists of:
\begin{enumerate}
\item for every $\XX : \BB_1$, a function $\hat{F}_0 : \DD_1(\XX) \to \DD_2(F\,\XX)$;
\item for every $f : \BB_1(\XX, \YY)$ and $\xxx : \DD_1(\XX), \yyy : \DD_1(\YY)$, a function 
$\hat{F}_1 : (\XX \to_\ff \YY) \to (\hat{F}_0\,\xxx \to_{F\,\ff} \hat{F}_0\,\yyy)$;
\end{enumerate}
The assignment $(\XX, \xxx) \to (F\,\XX, \hat{F}\, \xxx)$ bundles into a pseudofunctor $\int_F \hat{F} : \int \DD_1 \to \int \DD_2$, and is referred to as the \dfn{total pseudofunctor}.
\end{definition}
For the full definition of displayed pseudofunctors, that is, without contractible $2$-cells, see \cite{AFMVW2021} or \coqident{Bicategories.DisplayedBicats.DispPseudofunctor}{disp_psfunctor}.

\begin{example}[Displayed pseudofunctor from $\univ{\DD}$ to $\DD$]
The inclusion pseudofunctor $\iotaiota$ from $\int \univ{\DD}$ to $\int \DD$ is a total pseudofunctor over $\iota$, where the displayed pseudofunctor, denoted $\jota$, acts trivially on the displayed objects and morphisms.
\end{example}

In \cref{sec:on-rezk-for-top}, we construct a left biadjoint to $\iotaiota$ for concrete $\DD$.
Analogous to the construction of pseudofunctors between total bicategories, biadjoints can be constructed in two steps.
That is, we first construct a biadjoint between the base bicategories, and then do a construction at the level of the displayed bicategories.
The ingredients used to express the construction at the displayed level are bundled into the definition of \dfn{displayed universal arrows} \cref{dfn:disp_left_universal_arrow}.

For the definition of displayed adjoint equivalence between ($1$-)categories, we refer the reader to the formalization \coqident{CategoryTheory.DisplayedCats.Equivalences}{is_equiv_over} (also see pages $8$-$9$ in \cite{AL2019}, and \cite{AFMVW2021} for such equivalences between displayed bicategories).
When the base biadjoint is given by $\RC$, and when considering displayed bicategories with contractible $2$-cells, the conditions of being a displayed adjoint equivalence become particularly simple (see \cref{sec:disp-RC}). Hence, we will not recall this notion.

\begin{definition}%
[\coqident{Bicategories.DisplayedBicats.DisplayedUniversalArrow}{disp_left_universal_arrow}]
\label{dfn:disp_left_universal_arrow}
Let $R$ be a pseudofunctor from $\BB_1$ to $\BB_2$ with a left biadjoint $(L, \eta)$.
Let $\hat{R}$ be a displayed pseudofunctor from $\DD_1$ to $\DD_2$, over $R$.
A \dfn{(left) displayed universal arrow} for $\hat{R}$ (over $(R, L, \eta)$) consists of:
\begin{enumerate}
\item $\hat{L} : \prod_{\XX : \BB_2} \DD_2(\XX) \to \DD_1(L\,\XX)$, we write $\hat{L}(\xxx) := \hat{L}(\XX,\xxx)$;
\item a family $\hat{\eta}_{\xxx} : \xxx \rightarrow_{\eta\,\XX} \hat{R}(\hat{L}(\xxx))$ of displayed morphisms, for all $(\XX, \xxx)$ in $\int \DD_2$;
\end{enumerate}
and such that the displayed functor between displayed hom-categories
\[
\hat{\eta}_{\xxx} \cdot \hat{R}(-) : \DD_1(\hat{L}\,\xxx, \yyy) \to \DD_2(\xxx, \hat{R}\,\yyy),
\]
is a displayed adjoint equivalence whose base of displayment is the adjoint equivalence $\eta_\XX \cdot R(-)$ given by $(R,L,\eta)$.
\end{definition}

\begin{proposition}
[\coqident{Bicategories.DisplayedBicats.DisplayedUniversalArrow}{total_left_universal_arrow}]
\label{prop:total-biadjunction}
Let $(\hat{R}, \hat{L}, \hat{\eta})$ be a displayed universal arrow for $(R,L,\eta)$.
Then the total pseudofunctor
$\int_R \hat{R} : \int_{\BB_1} \DD_1 \to \int_{\BB_2} \DD_2$,
has a left biadjoint given by $(\XX, \xxx) \mapsto (L\,\XX, \hat{L}\,\xxx)$ and whose unit is given by $(\eta_\XX, \hat{\eta}_{\xxx})$.
\end{proposition}

\subsection{Displayed Universal Arrows over the Rezk Completion}
\label{sec:disp-RC}

In \cref{sec:disp-biadj}, we defined the notion of displayed universal arrows, which, in particular, provides the right formulation for the modular construction of biadjoints.
Often, there are simplifications due to how the bicategories are defined, for instance, in the case of topoi, the structure related to $2$-cells becomes irrelevant.
Furthermore, $\RC$ is completely characterized through weak equivalences, which allows for a further simplification.
In this section, we specialize the notion of displayed universal arrows to take into account these considerations.
In particular, we conclude that for lifting the Rezk completion to a bicategory such as $\TOP$, the displayed universal arrows become particularly simple (\cref{prop:disp-univ-arrow-over-cat}).

Recall that structure on a category is represented via a displayed bicategory $\DD$ over $\CAT$, whose displayed objects witness the structure and whose displayed morphisms witness preservation of the structure.
Then the bicategory of such structured categories is the total bicategory $\int_\CAT \DD$.
Furthermore, its full subbicategory on univalent structured categories is $\int_{\univ{\CAT}} \univ{\DD}$.
Moreover, we have the following diagram:

\begin{center}
\label{eqn:diagram:lifted-biadj}
\begin{tikzcd}
\int {\univ{\DD}}
	\arrow[d, swap, "U"]
	\arrow[rr, bend right=22, "\int \jota"] 
&& {\int \DD} \arrow[d, "U"] \arrow[ll, dashed, bend right=22, "?", swap, pos=0.6] \\
\univ{\CAT}
\arrow[rr, bend right=22, "\iota"] 
&& \CAT \arrow[ll, bend right=22, "\RC", pos=0.7]
\end{tikzcd}
\end{center}

Hence, the construction of the Rezk completion for categories with a $\DD$-structure corresponds to the construction of a left biadjoint of the inclusion $\iotaiota = \int \jota$.
If such a left biadjoint exists, we say that the Rezk completion (for categories) \emph{lifts} to categories with structure.

By \cref{prop:total-biadjunction}, the construction of such a left biadjoint amounts to the construction of a displayed universal arrow for $\RC$.
That is, for every $\CC \in \CAT$, and every displayed object $x : \DD(\CC)$ on $\CC$ (i.e., $\CC$ has a given $\DD$-structure), we have to construct a $\DD$-structure on $\RC(\CC)$ and show that $\eta_\CC : \CC \to \RC(\CC)$ is $\DD$-structure preserving. 
How one can construct those in concrete cases is explained in \cref{sec:on-rezk-for-top}.
Furthermore, we have to show that the displayed precomposition functor $\hat{\eta}_{\xxx} \cdot \hat{R}(-)$, see \cref{dfn:disp_left_universal_arrow}, is indeed a displayed adjoint equivalence. 
In a bicategory such as $\TOP$, this equivalence condition reduces to the statement that if a structure-preserving functor $F : \CC \to \EE$, into a univalent category, factorizes through $\eta_\CC$, then the functor $\RC(\CC) \to \EE$ is also structure-preserving.
More generally:

\begin{proposition}
[\coqident{Bicategories.DisplayedBicats.DisplayedUniversalArrowOnCat}{make_disp_left_universal_arrow_if_contr_CAT_from_weak_equiv}]
\label{prop:disp-univ-arrow-over-cat}
Let $\DD$ be a displayed bicategory over $\CAT$ such that for every weak equivalence $G : \CC_0 \to \CC_1$, whose codomain is univalent, we have:
\begin{enumerate}
\item for $\xxx : \DD(\CC_0)$, there is a $\hat{x} : \DD(\CC_1)$ and a $\hat{G} : \xxx \to_G \hat{x}$;
\item for every univalent category $\CC_2$, functors $F : \CC_0 \to \CC_2$ and $H : \CC_1 \to \CC_2$, and a natural isomorphism $\alpha : G \cdot H \Rightarrow F$, if $\xxx_i : \DD(\CC_i)$ and $\hat{F} : \xxx_0 \to_F \xxx_2$, then there is a $\hat{H} : \xxx_1 \to_H \xxx_2$.
\end{enumerate}
Then the pseudofunctor $\RC : \CAT \to \univ{\CAT}$ lifts to a left biadjoint for $\int \univ{\DD} \to \int \DD$.
\end{proposition}

\section{Examples}
\label{sec:on-rezk-for-top}

In this section, we lift the Rezk completion from categories to various categorical structures, resulting in the Rezk completion for (W-)topoi.
To construct the lifting, we use \cref{prop:disp-univ-arrow-over-cat}.
 
First, we consider categories equipped with finite limits in \cref{sec:RC-completeness}.
The lifting thereof is concluded in \cref{examples:limit-adjoint}.
Dually, we conclude the lifting of finite colimits in \cref{examples:colim-biadj}.
Furthermore, we discuss the lifting of infinite limits in \cref{lemma:AC-RC-inf}.

Afterwards, we consider elementary topoi.
In particular, we lift the Rezk completion to cartesian closed categories and categories with a subobject classifier.
The lifting thereof is concluded in \cref{examples:CCC-adjoint} and \cref{examples:eltopoi-adjoint} respectively.
Furthermore, in \cref{sec:RC:TOPN}, we consider the case where a category is equipped with a parameterized natural numbers object, which combined with the previous result, provides the Rezk completion for $W$-topoi.

\subsection{Finitely Complete Categories}
\label{sec:RC-completeness}

We now lift the Rezk completion from categories to finitely (co)complete categories.
It is well-known that the existence of finite limits follows from the existence of finite products and equalizers.
Hence, it is sufficient to lift the Rezk completion for those limits.
We discuss the lifting of binary products and refer the reader to the formalization for the lifting of terminal objects and equalizers.

\begin{definition}
[\coqident{Bicategories.DisplayedBicats.Examples.CategoriesWithStructure.BinProducts}{disp_bicat_binproducts}]
\label{dfn:catswithprod}
Let $\DD^{\times}$ be the $\CAT$-displayed bicategory whose displayed objects over $\CC$ witness that $\CC$ is equipped with binary products and whose displayed morphisms over a functor $F : \CC \to \DD$ witness that $F$ preserves binary products up to isomorphism.
The total bicategory, denoted $\CAT^\times$, is the bicategory of categories with binary products and product-preserving functors.
\end{definition}

Hence, the goal of this subsection is to lift $\RC : \CAT \to \univ{\CAT}$ to a left biadjoint for $\univ{\CAT^\times} \hookrightarrow  \CAT^\times$.
By \cref{prop:disp-univ-arrow-over-cat}, it suffices to prove:
\begin{theorem}
\label{thm:RC-for-prod}
Let $\CC$ be a category with binary products.
Then
\begin{enumerate}
\item\label{RC:prod:create} the Rezk completion $\RC(\CC)$ of $\CC$ is equipped with binary products;
\item\label{RC:prod:preserve} the weak equivalence $\eta_\CC : \CC \to \RC(\CC)$ preserves binary products.
\end{enumerate}
Furthermore,
let $\EE$ be a univalent category with binary products and $F : \CC \to \EE$ a functor preserving binary products up to isomorphism.
Suppose $F$ factorizes through $\eta_\CC$ as depicted in the following diagram.
\begin{equation}
\label{eqn:goal_pb_liftpreservation-alpha}
\begin{tikzcd}
  {\CC} \arrow[rr, "\eta_\CC"] & & {\RC(\CC)} \arrow[ld, "H"]\\
  & {\EE} &
  \arrow[""{name=0, anchor=center, inner sep=0}, swap, "{F}", from=1-1, to=2-2]
  \arrow["\alpha", shorten <=12pt, shorten >=12pt, Leftrightarrow, from=0, to=1-3]
\end{tikzcd}
\end{equation}
Then $H$ preserves binary products.
\end{theorem}

The proof of \cref{thm:RC-for-prod} happens in two steps.
First, we show in \cref{prop:RC-has-prod} that \cref{RC:prod:create} follows from the statement that weak equivalences preserve binary products (\cref{prop:weq-preserve-binprod}) which is precisely \cref{RC:prod:preserve}.
That is, the image of a product is a product.
Second, we show in \cref{lemma:lift-preservation-prod} that the last statement follows from the statement that weak equivalences reflect binary products.

We write products for binary products.

\begin{subparagraph}{The creation of products}
We now show that if $\CC$ is a category with products, then $\RC(\CC)$ has products and those are preserved by $\eta_\CC$.
First, we show that weak equivalences preserve products.
\begin{proposition}
[\coqident{CategoryTheory.WeakEquivalences.Preservation.Binproducts}{weak_equiv_preserves_binproducts}]
\label{prop:weq-preserve-binprod}
Let $G : \CC \to \DD$ be a weak equivalence and $(x_1 \xrightarrow{\pi_1} p \xleftarrow{\pi_2} x_2)$ a product in $\CC$.
Then $(G\,x_1 \xleftarrow{G\,\pi_1} G\,p \xrightarrow{G\,\pi_2} G\,x_2)$ is a product in $\DD$.
\end{proposition}

\begin{proof}
Let $(G\,x_1 \xrightarrow{g_1} y \xleftarrow{g_2} G\,x_2)$ be a cone in $\DD$.
We have to construct a morphism $h : \DD(y, G\,p)$ which is unique w.r.t., to the property that for $k = 1,2$, the following diagram commutes:
\begin{equation}
\label{eqn:weq-preserve-binprod-proj}
\begin{tikzcd}
& y \arrow[d, "h"] \arrow[dl, "g_k", swap] \\
{G\,x_k} & {G\,p} \arrow[l, "G\,\pi_k"] 
\end{tikzcd}
\end{equation}

Type-theoretically, we need to show that the following type is contractible:
\begin{equation}
\label{eqn:type-witnessing-h}
A := \sum\, h : \DD(y, G\,p),\, h \cdot G(\pi_1) = g_1 \times h \cdot G(\pi_2) = g_2.
\end{equation}
To show that $A$ is inhabited, we use essential surjectivity of $G$ to obtain an $x : \CC$ and an isomorphism $i : G\,x \cong y$.
Note that we can construct $x$ and $i$, because being contractible is a proposition.

We claim that $h$ is given by $y \xrightarrow{i^{-1}} {G\,x} \xrightarrow{G(\langle f_1, f_2 \rangle)} G\,p$, where $f_k := G^{-1}(i \cdot g_k)$ and $\langle f_1, f_2 \rangle$ is the unique morphism given by the universal property of $p$.

That $h$ is an inhabitant of $A$ follows from the following calculation (\coqident{CategoryTheory.WeakEquivalences.Preservation.Binproducts}{weak_equiv_preserves_binproducts_pr1}):
\begin{eqnarray*}
h \cdot G(\pi_k) =& i^{-1} \cdot G(\langle f_1, f_2 \rangle) \cdot G(\pi_k),
	&\quad \text{ by dfn. of $h$}, \\
=& i^{-1} \cdot G(\langle f_1, f_2 \rangle \cdot \pi_k), 
	&\quad \text{ by functoriality of $G$},\\
=& i^{-1} \cdot G(f_k) , 
	&\quad \text{by dfn. of $\langle f_1, f_2 \rangle$} \\ 
=& i^{-1} \cdot G(G^{-1}(i \cdot g_k)) = g_1, 
	&\quad \text{ by dfn. of $f_k$}.\\
\end{eqnarray*}

It remains to verify that the type $A$ is a proposition (\coqident{CategoryTheory.WeakEquivalences.Preservation.Binproducts}{weak_equiv_preserves_binproducts_unique}).
Let $\phi_1 := (h_1, q_1)$ and $\phi_2 := (h_2, q_2)$ be in $A$.
Since $A$ is a subtype of $\DD(y, G\,p)$, it suffices to show that $h_1 = h_2$.
This is equivalent to showing $G^{-1}(i \cdot h_1) = G^{-1}(i \cdot h_2)$ since $i$ is an isomorphism and a fully faithful functor reflects equality of morphisms.
The result now follows by the universal property of the product and by assumption we have $h_1 \cdot G(\pi_k) = h_2 \cdot G(\pi_k)$.
\end{proof}

Observe that \cref{prop:weq-preserve-binprod} entails that if $\CC$ and $\DD$ are equipped with binary products, and if $\DD$ is univalent (i.e., $\DD = \RC(\CC)$), then $G(x \times_\CC y) = G(x) \times_\DD G(y)$.

Before concluding that $\RC(\CC)$ is equipped with products if $\CC$ is, observe:
\begin{lemma}
[\coqident{CategoryTheory.Limits.BinProducts}{isaprop_BinProduct}]
In a univalent category, products are unique up to identity, if they exist.
\end{lemma}

\begin{proposition}
[\coqident{CategoryTheory.WeakEquivalences.Creation.BinProducts}{weak_equiv_into_univ_creates_binproducts}]
\label{prop:RC-has-prod}
If $\CC$ is equipped with binary products, then $\RC(\CC)$ is equipped with binary products.
\end{proposition}
\begin{proof}
Let $y_1, y_2 : \RC(\CC)$.
We have to construct the binary product of $y_1$ and $y_2$.
By univalence of $\RC(\CC)$, the product of $y_1$ and $y_2$ is necessarily unique, if it exists.
Hence, we can use essential surjectivity of $\eta_\CC$ to construct $x_1$ and $x_2$ such that $\eta_\CC(x_k) \cong y_k$.
Since the product is closed under isomorphism, it suffices to construct the product of $G(x_1)$ and $G(x_2)$, which we have constructed in \cref{prop:weq-preserve-binprod}.
\end{proof}
\end{subparagraph}

\begin{subparagraph}{The universality}
The final step in lifting $\RC$ to categories with binary products is the furthermore clause in \cref{thm:RC-for-prod}.
Thus, consider the following diagram of categories:
\begin{equation}
\label{eqn:goal_pb_liftpreservation-alpha}
\begin{tikzcd}
  {\CC} \arrow[rr, "G"] & & {\DD} \arrow[ld, "H"]\\
  & {\EE} &
  \arrow[""{name=0, anchor=center, inner sep=0}, swap, "{F}", from=1-1, to=2-2]
  \arrow["\alpha", shorten <=12pt, shorten >=12pt, Leftrightarrow, from=0, to=1-3]
\end{tikzcd}
\end{equation}
We have to show that if $G$ is a weak equivalence and if $F$ preserves binary products, then $H$ preserves binary products.

\begin{lemma}
[\coqident{CategoryTheory.WeakEquivalences.Reflection.BinProducts}{weak_equiv_reflects_products}]
\label{lemma:weq-reflect-prod}
A fully faithful functor reflects products.
\end{lemma}

\begin{lemma}
[\coqident{CategoryTheory.WeakEquivalences.LiftPreservation.BinProducts}{weak_equiv_lifts_preserves_binproducts}]
\label{lemma:lift-preservation-prod}
Assume that $G$ is a weak equivalence and that $F$ preserves binary products.
Then $H$ preserves binary products.
\end{lemma}
\begin{proof}
We have to show that if $(y_1 \xrightarrow{\pi_1} p_y \xleftarrow{\pi_2} y_2)$ is a product cone in $\DD$, then $(H\,y_1 \xrightarrow{H\,\pi_1} H\,p_y \xleftarrow{H(\pi_2)} H\,y_2)$ is a product in $\EE$.

Since $G$ is essential surjective, we have objects $x_1, x_2, p_x : \CC$ such that $G\,x_1 \cong y_1, G\,x_2 \cong y_2$ and $G\,p_x \cong p_y$.
Note that we can construct such objects and isomorphisms, because being a product is a proposition.
Now, since $\alpha$ is a natural isomorphism, we have isomorphisms $H(y_k) \cong F(x_k)$, for $k = 1,2$, and $H(p_y) \cong F(p_x)$.

By \cref{lemma:weq-reflect-prod}, $p_x$ is a product of $x_1$ and $x_2$.
Then by assumption on $F$, $F(p_x)$ is a product of $F(x_1)$ and $F(x_2)$.
The claim now follows since products are closed under isomorphism and because the following diagram commutes:
\[
\begin{tikzcd}
F(x_1) \arrow[r, "F(\pi_1)"] \arrow[d, "\cong"] & F(p_x) \arrow[d, "\cong"] & F(x_2) \arrow[l, "F(\pi_2)", swap] \arrow[d, "\cong"] \\
H(y_1) \arrow[r, swap, "H(\pi_1)"] & H(p_y) & H(y_2) \arrow[l, "H(\pi_2)"]
\end{tikzcd}
\]
\end{proof}

Hence, \cref{lemma:lift-preservation-prod} implies that the action of $\RC : \CAT \to \univ{\CAT}$ on morphisms (functors) preserves the preservation of products.
This concludes the proof of \cref{thm:RC-for-prod}.

\end{subparagraph}

\begin{subparagraph}{The lifted biadjunction}
We now conclude that the Rezk completion for categories lifts to categories with binary products.
Furthermore,  we conclude the same result for finite limits more generally, or more precisely: terminal objects and equalizers. In the formalization, we also consider pullbacks, because this is used for lifting the subobject classifier.

\begin{proposition}
[\coqident{Bicategories.RezkCompletions.StructuredCats.BinProducts}{disp_bicat_chosen_binproducts_has_Rezk_completions}]
\label{examples:pb-adjoint}
The inclusion $\univ{(\CAT^\times)} \hookrightarrow \CAT^\times$ has a left biadjoint.
\end{proposition}
\begin{proof}
The result follows from \cref{prop:disp-univ-arrow-over-cat} whose assumptions are proven in \cref{lemma:lift-preservation-prod,prop:RC-has-prod}.
\end{proof}

\begin{definition}
[\coqident{Bicategories.DisplayedBicats.Examples.CategoriesWithStructure.FiniteLimits}{disp_bicat_limits}]
\label{dfn:catswithlimits}
Let $\DD^{\mathsf{lim}}$ be the product of the $\DD^\times$ and the displayed bicategories $\DD^\mathsf{1}$ and $\DD^{\mathsf{eq}}$ whose displayed objects over a category $\CC$ witness that $\CC$ is equipped with a terminal object, resp.\, equalizers, and whose displayed morphisms over a functor witness the preservation of the terminal object and equalizers respectively.
The total bicategory of $\DD^{\mathsf{lim}}$, denoted $\LCAT$, is the bicategory of finitely complete categories and finite-limit preserving functors.
\end{definition}

\begin{proposition}%
[\coqident{Bicategories.RezkCompletions.StructuredCats.FiniteLimits}{disp_bicat_chosen_limits_has_RezkCompletion}]
\label{examples:limit-adjoint}
The inclusion $\univ{\LCAT} \hookrightarrow \LCAT$ has a left biadjoint.
\end{proposition}
\end{subparagraph}

\begin{subparagraph}{Colimits}
Previously, we've constructed the Rezk completion for finitely complete categories.
The duality between limits and colimits entail the same result for finitely cocomplete categories, as a consequence of the following lemma:

\begin{lemma}
[\coqident{CategoryTheory.WeakEquivalences.Opp}{opp_is_weak_equiv}]
Let $F : \CC \to \DD$ be a functor.
Then $F$ is a weak equivalence if and only $\op{F} : \op{\CC} \to \op{\DD}$ is a weak equivalence.
\end{lemma}

We define the bicategory $\CLCAT$ of finitely cocomplete categories similarly to the bicategory of finitely complete categories.
That is, as the product of displayed bicategories of binary coproducts and coequalizers \coqident{Bicategories.DisplayedBicats.Examples.CategoriesWithStructure.FiniteColimits}{disp_bicat_colimits}:

\begin{corollary}
[\coqident{Bicategories.RezkCompletions.StructuredCats.FiniteColimits}{disp_bicat_colimits_has_RezkCompletion}]
\label{examples:colim-biadj}
The inclusion $\univ{\CLCAT} \hookrightarrow \CLCAT$ has a left biadjoint.
\end{corollary}
\end{subparagraph}

\begin{subparagraph}{Infinite limits}

Whereas the Rezk completion lifts to finitely complete categories, the same result does not need to hold for infinitary limits.

Recall that the axiom of $(n,m)$-choice denotes the statement \cite[Exer. 7.8.]{hottbook}
\[
\prod_{X : \SET} \prod_{Y : X \to \UU_n} \left(\prod_{x : X} \trunc{Y(x)}_m\right) \to \trunc{\prod_{x : X} Y(x)}_m,
\]
where $-1 \leq n,m \leq \infty$ and $\UU_n$ denotes the type of $n$-types.

\begin{lemma}
\label{lemma:AC-RC-inf}
Assume the axiom of $(\infty,-1)$-choice holds.
Let $\CC$ be a category equipped with limits indexed by a setcategory $\mathcal{A}$.
Then $\RC(\CC)$ has all $\mathcal{A}$-indexed limits.
\end{lemma}
\begin{proof}
Let $\mathcal{J} : \mathcal{A} \to \RC(\CC)$ be a diagram.
We want to conclude that $\mathcal{J}$ has a limit.
Let $X := \ob{\mathcal{A}}$ and $Y : \mathcal{A} \to \UU_\infty : a \mapsto \sum_{x : \CC} \mathcal{J}(a) \cong \eta_\CC(x)$.
Observe that $\mathcal{J}(a) \cong \eta_\CC(x)$ is a set, but that the homotopy level of $\ob{\CC}$ might indeed be infinite.

By choice, we have a function:
\[
\left(\prod_{a : \mathcal{A}} \trunc{\sum_{x : \CC} \mathcal{J}(a) \cong \eta_\CC(x)}\right) \to \trunc{\prod_{a : \mathcal{A}} \sum_{x : \CC} \mathcal{J}(a) \cong \eta_\CC(x)},
\]
Since $\eta_\CC$ is essentially surjective, the domain thereof is inhabited.
Hence, we are left to construct
\[
\trunc{\prod_{a : \mathcal{A}} \sum_{x : \CC} \mathcal{J}(a) \cong \eta_\CC(x)} \to \lim(\mathcal{J}),
\]
where $\lim(\mathcal{J})$ is the type of limits of $\mathcal{J}$.
Since $\lim(\mathcal{J})$ is a proposition by univalence of $\RC(\CC)$, we are given 
\[
k' : \prod_{a : \mathcal{A}} \sum_{x : \CC} \mathcal{J}(a) \cong \eta_\CC(x).
\]
In particular, we have a diagram $\mathcal{K} : \mathcal{A} \to \CC$, which by assumption on $\CC$ has a limit cone $\lim(\mathcal{K})$.
We claim that $\eta_\CC(\lim(\mathcal{K}))$ is a limit for $\mathcal{J}$.

First, $\eta_\CC$ preserves limits.
Hence, $\eta_\CC(\lim(\mathcal{K}))$ is a limit for the diagram $\mathcal{K} \cdot \eta_\CC$.
Since limits are closed under isomorphic diagrams, the result follows if $\mathcal{J} \cong \mathcal{K} \cdot \eta_\CC$.
Such an isomorphism is indeed given by $k'$.

Observe that the naturality and functoriality conditions need not be verified since we can assume that $\mathcal{J}$ is discrete.
\end{proof}

\end{subparagraph}

\subsection{Cartesian Closed Categories}
\label{sec:RC-CCC}

We now lift the Rezk completion to cartesian closed categories.
The approach taken therefore is the same as for the lifting of limits.
In particular, we rely on the results from \cref{sec:RC-completeness}.

Recall that an object $x : \CC$ in a category with binary products is \dfn{exponentiable} if $- \times x : \CC \to \CC$ has a right adjoint, which is denoted as $\expo{x}{(-)} : \CC \to \CC$.
That is, for every $y : \CC$ there are given an object $\expo{x}{y} : \CC$ and a morphism $\ev := \ev_{x,y} : \CC(\expo{x}{y} \times x, y)$ which is universal in the sense that $(\expo{x}{y}, \ev)$ is the terminal such pair.
The object $\expo{x}{y}$ is referred to as the \textit{exponential} of $x$ with $y$, and $\ev_{x,y}$ is referred to as the \textit{evaluation morphism}.
We say that a category is \dfn{cartesian closed} if it has all finite products and all exponentials.

Furthermore, recall that a binary product preserving functor $F : \CC \to \EE$ between cartesian closed categories is a \dfn{cartesian closed functor} which \dfn{preserves exponentials}.
That is, for every $x,y : \CC$ the unique morphism from $F(\expo{x}{y})$ to $\expo{F(x)}{F(y)}$ is an isomorphism.

Recall from \cref{dfn:catswithprod} that we denote by $\CAT^\times$ the bicategory of categories with binary products.
\begin{definition}
[\coqident{Bicategories.DisplayedBicats.Examples.CategoriesWithStructure.Exponentials}{disp_bicat_exponentials}]
Let $\DD^{\mathsf{exp}}$ be the $\CAT^\times$-displayed bicategory whose displayed objects over $CC$ witness that $\CC$ is cartesian closed and whose displayed morphisms witness the preservation of exponentials.
The total bicategory, denoted $\CCC$, is the bicategory of cartesian closed categories and cartesian closed functors.
\end{definition}

Hence, to lift the Rezk completion to cartesian closed categories, we lift $\RC : \CAT^\times \to \univ{\CAT}^\times$ to a left biadjoint for $\univ{\CCC} \hookrightarrow \CCC$.
Analogous to before, we use \cref{prop:disp-univ-arrow-over-cat} and take the same steps as before.

Recall that the Rezk completion lifts to categories with binary products.
In particular, there is a necessarily unique isomorphism $\mu := \mu_{x,y} : \eta(x) \times \eta(y) \cong \eta(x \times y)$, for every $x,y : \CC$ and natural in $x$ and $y$.
\begin{lemma}
[\coqident{CategoryTheory.WeakEquivalences.Preservation.Exponentials}{weak_equiv_preserves_exponentiable_objects_uvp}]
\label{examples:CCC-image}
Let $\CC$ be a cartesian closed category and $G : \CC \to \EE$ a weak equivalence.
Then the image of an exponential, under $G$ is again an exponential.
That is, if $(\expo{x}{y}, ev)$ is an exponential of $x$ with $y$, then $\left(G(\expo{x}{y}), G(\ev)\right)$ is an exponential of $G(x)$ with $G(y)$.
\end{lemma}

\begin{proposition}
[\coqident{CategoryTheory.WeakEquivalences.Preservation.Exponentials}{weak_equiv_into_univ_creates_exponentials}
\coqident{CategoryTheory.WeakEquivalences.Preservation.Exponentials}{weak_equiv_preserves_exponentials}]
\label{prop:RC-has-exp}
If $\CC$ is a cartesian closed category, then $\RC(\CC)$ is cartesian closed.
Furthermore, $\eta_\CC : \CC \to \RC(\CC)$ is a cartesian closed functor.
\end{proposition}
\begin{proof}
Since exponentials are unique up to isomorphism, the preservation of exponentials is equivalent to the statement that the image of an exponential object is again an exponential object as in \cref{examples:CCC-image}.
\end{proof}

We are now able to conclude that the Rezk completion lifts from categories to cartesian closed categories.

\begin{proposition}
[\coqident{Bicategories.RezkCompletions.StructuredCats.FiniteLimits}{disp_bicat_chosen_limits_has_RezkCompletion}]
\label{examples:CCC-adjoint}
The inclusion $\univ{\CCC} \hookrightarrow \CCC$ has a left biadjoint.
\end{proposition}
\begin{proof}
Again, we use \cref{prop:disp-univ-arrow-over-cat} to $\DD := \DD^\mathsf{exp}$.
The first two conditions thereof is precisely \cref{prop:RC-has-exp}.
The last condition (universality) follows since weak equivalences reflect exponential objects.
\end{proof}

\subsection{Elementary Topoi}
\label{sec:RC-TOP}

We now lift the Rezk completion to elementary topoi: Cartesian closed categories equipped with equalizers and a subobject classifier.
Observe that elementary topoi are necessarily finitely complete.
By \cref{examples:CCC-adjoint,examples:limit-adjoint}, we know that the cartesian closed structure and the equalizers are lifted.
Hence, it remains to lift the subobject classifier.

Let $\CC$ be a category with a terminal object $T$.
Recall that a \dfn{subobject classifier} consists of an object $\Omega : \CC$ and a monomorphism $\tau : \CC(T, \Omega)$, which furthermore satisfies the following universal property.
For every mono $f : \CC(x, y)$, there is a unique morphism $\chi_f : \CC(y, \Omega)$ such that the following diagram commutes and is a pullback square:
\[
\begin{tikzcd}
x \arrow[r, "!"] \arrow[d, swap, "f"] & T \arrow[d, "\tau"] \\
y \arrow[r, swap, "\chi_{f}"] & {\Omega}
\end{tikzcd}
\]
We refer to $\tau$ as the \textit{truth map}.

Furthermore, recall that a functor $F : \CC_0 \to \CC_1$ between categories with a terminal object, denoted $T_i$ for $i=0,1$, and a subobject classifier $(\Omega_i, \tau_i)$ preserves the subobject classifier if $F$ preserves the terminal object and one of the following two equivalent conditions holds:
\begin{enumerate}
\item $\left( F(\Omega_0) , T_1 \xrightarrow{!} F(T_0) \xrightarrow{F(\tau_0)} F(\Omega_0)\right)$ is a subobject classifier in $\CC_1$;
\item there is an isomorphism $i : F(\Omega_0) \simeq \Omega_1$ such that $F(\tau_0) \cdot i =\, !^{-1} \cdot \tau_1$;
\end{enumerate}
where $!$ denotes the unique morphism given by terminality of $F(T_0)$.

Recall from \cref{dfn:catswithlimits} that we denote by $\CAT^\mathsf{1}$ the bicategory of categories with a terminal object.
\begin{definition}
[\coqident{Bicategories.DisplayedBicats.Examples.CategoriesWithStructure.SubobjectClassifier}{disp_bicat_subobject_classifier}]
Let $\DD^{\Omega}$ be the $\CAT^\mathsf{1}$-displayed bicategory whose displayed objects over a category with a terminal object $\CC$ witness that $\CC$ is equipped with a subobject classifier and whose displayed morphisms witness the preservation of the subobject classifier.
The total bicategory, denoted $\CAT^\Omega$, is the bicategory of categories with a terminal object and a subobject classifier and whose morphisms are functors preserving the terminal object and the subobject classifier.
\end{definition}

Recall that we denote by $\LCAT$ the bicategory of finitely complete categories, which is constructed as $\LCAT := \int_\CAT \DD^\mathsf{lim}$.
Furthermore, observe that $\DD^\mathsf{exp}$ and $\DD^\Omega$ can be considered as displayed bicategories over $\LCAT$.
\begin{definition}
[\coqident{Bicategories.DisplayedBicats.Examples.CategoriesWithStructure.Topoi}{disp_bicat_elementarytopoi}]
Let $\DD^\mathsf{top}$ be the product of the $\LCAT$-displayed bicategories $\DD^\mathsf{exp}$ with $\DD^\Omega$.
The total bicategory of $\DD^\mathsf{top}$, denoted $\TOP$, is the bicategory of elementary topoi, logical functors, and natural transformations.
\end{definition}

Recall that a morphism $f : \CC(x,y)$ is a monomorphism if and only if 
\[
\begin{tikzcd}
x \arrow[r, equal] \arrow[d, equal] & x \arrow[d, "f"] \\
x \arrow[r, swap, "f"] & y
\end{tikzcd}
\]
is a pullback square.
\begin{lemma}
[\coqident{CategoryTheory.WeakEquivalences.Mono}{weak_equiv_preserves_monic}\coqident{CategoryTheory.WeakEquivalences.Mono}{reflection_of_mono_is_mono}]
\label{lemma:mono-weq}
Weak equivalences preserve and reflect monomorphisms.
\end{lemma}
Observe that the reflection of monomorphisms follows from only fully faithfulness, but not essential surjectivity.

Recall that if $G$ is a weak equivalence and $T_0$ is a terminal object, then $G(T_0)$ is terminal.
\cref{examples:omega-image} follows since weak equivalences preserve pullbacks:
\begin{lemma}
[\coqident{CategoryTheory.WeakEquivalences.Preservation.SubobjectClassifier}{weak_equiv_preserves_subobject_classifier}]
\label{examples:omega-image}
Let $G : \CC \to \EE$ be a weak equivalence between categories with a terminal object, denoted $T_\CC$ and $T_\EE$ respectively.
Suppose that $(\Omega, \tau : T_\CC \to \Omega)$ is a subobject classifier in $\CC$.
Then $G(\Omega)$ is a subobject classifier whose \textit{truth map} is $! \cdot G(\tau) : T_\EE \to G(\Omega)$, where $!$ is the unique morphism from $T_\EE$ to $G(T_\CC)$.
\end{lemma}

\begin{proposition}
[\coqident{Bicategories.RezkCompletions.StructuredCats.SubobjectClassifier}{disp_bicat_subobject_classifier_has_RezkCompletion}]
\label{examples:omega-adjoint}
The inclusion $\univ{\CAT}^\Omega \hookrightarrow \CAT^\Omega$ has a left biadjoint.
\end{proposition}

We now conclude that $\RC : \LCAT \to \univ{\LCAT}$ lifts to a left biadjoint for $\univ{\TOP} \to \TOP$.

\begin{theorem}
[\coqident{Bicategories.RezkCompletions.StructuredCats.Topoi}{disp_bicat_elementarytopoi_has_RezkCompletion}]
\label{examples:eltopoi-adjoint}
The inclusion $\univ{(\TOP)} \hookrightarrow \TOP$ has a left biadjoint $\RC^{top}$.
In particular:
\begin{enumerate}
\item the Rezk completion of an elementary topos $\EE$ is an elementary topos $\RC^{top}(\EE)$;
\item the weak equivalence $\eta_\EE : \EE \to \RC^{top}(\EE)$ is a logical functor;
\item and, $(\RC^{top}(\EE), \eta_\EE)$ is universal among the univalent elementary topoi.
\end{enumerate}
\end{theorem}

\subsection{W-Topoi}
\label{sec:RC:TOPN}

As a final example, we lift the Rezk completion to categories equipped with a (parameterized) natural numbers object.
In combination with \cref{examples:omega-adjoint}, we conclude that the Rezk completion lifts from categories to elementary topoi with a natural numbers object.
Since elementary topoi with a natural numbers object have W-types \cite[Prop. 3.6.]{MoerdijkPalmgren}, we call such topoi \dfn{W-topoi}.

The most common way to interpret, or axiomatize, the object of natural numbers in a category leads to the definition of a \textit{natural numbers object}.
Nonetheless, in the absence of exponentials, there is a more appropriate interpretation by strengthening the recursion principle and is known as a \textit{parameterized natural numbers object} \coqident{CategoryTheory.Arithmetic.ParameterizedNNO}{parameterized_NNO} \cite{M2010}.
We work with the more general version.

A \dfn{parameterized natural numbers object} (pNNO) in a category $\CC$ with finite products is a tuple $(\NNO, z, s)$ where $\NNO : \CC$ is an object and $z : \CC(T, \NNO)$, $s : \CC(\NNO,\NNO)$ are morphisms which satisfies the following universal property: for every tuple $(t : \CC, m : \CC, z' : \CC(t, m), s' : \CC(m,m))$, there exists a unique $f : \CC(t \times \NNO, m)$ such that the following diagram commutes:
\[
\begin{tikzcd}
{t} \arrow[r, "{\langle \id_t\, ,\, ! \cdot z \rangle}"] \arrow[rd, swap, "{z'}"] & {t \times \NNO} \arrow[d, "f"] & {t \times \NNO} \arrow[d, "f"] \arrow[l, swap, "{\id \times s}"] \\
& m & m \arrow[l, "{s'}"]
\end{tikzcd}
\]

Recall that the image of a terminal object under a weak equivalence is again terminal.
\begin{lemma}
[\coqident{CategoryTheory.WeakEquivalences.PNNO}{weak_equiv_preserves_parameterized_NNO}\coqident{CategoryTheory.WeakEquivalences.PNNO}{weak_equiv_reflects_parameterized_NNO}]
\label{examples:nno-preserve-and-reflect}
Let $G : \CC_0 \to \CC_1$ be a weak equivalence where both $\CC_0$ and $\CC_1$ have finite products.
Assume that $\NNO : \CC_0$, $z : \CC_0(T_0, \NNO)$, and $s : \CC_0(\NNO,\NNO)$ are given.
Then $(\NNO, z, s)$ is a pNNO if and only if $\left( G(\NNO), G(z), G(s) \right)$ is a pNNO.
\end{lemma}

Let $\NNOCAT$ be the bicategory whose objects are categories equipped with finite products and a pNNO, and whose morphisms are terminal preserving functors $F : \CC_0 \to \CC_1$ such that the unique morphism from $\NNO_1 \to F(\NNO_0)$, induced by the universal property of the pNNO, is an isomorphism.

\begin{proposition}
[\coqident{Bicategories.RezkCompletions.StructuredCats.ParameterizedNNO}{disp_bicat_parameterized_NNO_has_Rezk_completion}]
\label{examples:nno-adjoint}
The inclusion $\univ{(\NNOCAT)} \hookrightarrow \NNOCAT$ has a left biadjoint.
\end{proposition}

\begin{corollary}
[\coqident{Bicategories.RezkCompletions.StructuredCats.Topoi}{disp_bicat_arithmetic_elementarytopoi_has_RezkCompletion}]
\label{examples:topnno-adjoint}
The left biadjoint $\RC$ lifts to elementary topoi with an (parameterized) natural numbers object.
\end{corollary}

\section{Related Work}
\label{sec:related}

We elaborate on the different methods to extend the Rezk completion to categories with structure found in the literature.

In \cite{ALV2017}, the Rezk completion has been extended to, in particular, categories with families (CwFs).
That is, they showed that every weak equivalence into a univalent category (i.e., the Rezk completion) induces a map from the type of CwFs on the domain to the type of CwFs on the codomain.
Furthermore, they constructed for every weak equivalence an equivalence between the representable map structures on the involved categories.
These results allowed them to construct an equivalence between representable maps of presheaves on a category and the CwFs on its Rezk completion.
While our focus has been on displayed bicategories with contractible $2$-cells, CwFs and representable maps require propositional $2$-cells.
Nonetheless, \cref{prop:disp-univ-arrow-over-cat} could be slightly generalized to have propositional $2$-cells.
In particular, we expect our methodology to also work for these structures.

In \cite{VW2021}, it was proven that the inclusion of the bicategory of univalent groupoids into the bicategory of all groupoids admits a left biadjoint.
Furthermore, it was shown that this biadjoint lifts to a variety of structures on groupoids.
However, instead of considering each structure individually, as we have done, \cite{VW2021} presented a signature for higher inductive types to encode structure on groupoids and then showed that the structure definable by such HITs induces a left biadjoint for the inclusion of structured univalent groupoids into structured groupoids.
In particular, \cite{VW2021} considers \textit{algebraic} structure, whereas we consider structure defined via a universal property.
Furthermore, while we work with biadjoints in terms of equivalences between hom-categories, \cite{VW2021} works with biadjoints in terms of the unit-counit description.

In \cite{WMA2022}, the Rezk completion has been extended to \textit{monoidal categories}.
In particular, they also proved that the biadjoint given by the Rezk completion lifts from categories to monoidal categories, although not in the exact same words.
The approach considered here differs in the following way. 
First, we do not rely on the universal property of weak equivalences and the Rezk completion. 
Instead, we do the lifting in terms of weak equivalences. 
Second, whereas in \cite{WMA2022}, they relied on displayed categories to
lift the equivalence on hom-categories, we now rely on displayed bicategories to
formulate displayed biadjunctions/universal arrows.

Instead of lifting via \textit{abstract Rezk completions}, one can also try to use a \textit{concrete implementation} of the Rezk completion, either via the presheaf construction, or via higher inductive types.
In \cite{WMA2022}, it was furthermore sketched how the presheaf construction inherits a monoidal structure by considering the Day convolution structure on the category of presheaves.
In \cite{vdW24-enriched}, both implementations of the Rezk completion have been extended to the setting of enriched categories.
However, as opposed to the monoidal case, some modifications have been made.
First, the notion of fully faithfulness, and hence weak equivalence, has to be adapted to take the enrichment into account.
Nonetheless, \cite{vdW24-enriched} showed that such weak equivalences imply an equivalence on the \textit{enriched functor categories}, similar to the (non-enriched) weak equivalences.
Second, as opposed to the monoidal Rezk completion, \cite{vdW24-enriched} considered enriched presheaves, as opposed to $\SET$-valued presheaves.
Note that even though the enriched presheaf construction does a priori not correspond to the ordinary presheaf construction, the HIT construction for the enriched Rezk completion shows that the underlying category of the enriched Rezk completion coincides with the \textit{ordinary} Rezk completion.

\section{Conclusion}
\label{sec:conclusion}

\begin{subparagraph}{Conclusion}
We presented a modular approach for lifting the Rezk completion from categories to structured categories, based on displayed universal arrows (\cref{dfn:disp_left_universal_arrow}).
We illustrated this approach, and its modularity, by lifting the Rezk completion from categories to topoi (\cref{examples:topnno-adjoint}).
In particular, we've observed that the lifting of the Rezk completion essentially reduces to the preservation and reflection properties of weak equivalences.
Moreover, the approach presented here is independent of an explicit construction, but follows from its characterization in terms of weak equivalences.
However, while elementary topoi already illustrate that various structures do lift, \cref{lemma:AC-RC-inf} illustrates that for certain structures we need additional assumptions in our logical foundation.
Lastly, topoi are an important notion in categorical logic and semantics since they allow for an interpretation of various type formers and logical constructs.
Hence, our results imply that the Rezk completion preserves these operations, which means that the Rezk completion can be used for the construction of univalent models.
\end{subparagraph}

\begin{subparagraph}{Future Directions}
There are multiple ways in which this work can be extended.
First, in \cite{AFMVW2021}, it is proven that the presheaf construction for the Rezk completion of ($1$-)categories induces a similar result when passing to (locally univalent) bicategories.
Hence, there is the question of how the results presented here \textit{categorify} to higher categories and topoi (see \eg, \cite{S1982,W2007}).

Second, our results imply that the tripos-to-topos construction (\cref{exa:trip-to-top}) can be extended to a \textit{tripos-to-univalent-topos construction}.
Consequently, this can be used for the construction of realizability topoi.
However, it remains to be seen whether the resulting topoi satisfy the assumptions of \cite[Thm. 4.6]{F2019}.

Third, there are more categorical structures one can consider.
One such example is \textit{locally cartesian closedness}, used to interpret $\Pi$-types.
We expect this to follow from the results about cartesian closed categories presented above and the fact that taking slice categories is well-behaved with respect to weak equivalences.
Furthermore, the approach presented also works for other categorical structures.
For example, abelian categories can be characterized in an \textit{enrichment-free} way using (finitary) universal constructions, whose preservation properties should follow from the above.
\end{subparagraph}

\newpage
\bibliographystyle{plainurl}
\bibliography{main}

\end{document}